\documentclass[10pt]{article}      

\usepackage{amsmath,amssymb,url,listings,mathrsfs}
\usepackage{graphicx}
\usepackage{subfigure}
\usepackage{authblk}

\newcommand{\num}[1]{\relax\ifmmode \mathbb #1\else $\mathbb #1$\fi}
\newcommand{\naturals}{{\num N}}
\newcommand{\reals}{{\num R}}
\newcommand{\R}{{\mathcal{R}}}
\newcommand{\T}{{\mathcal{T}}}
\DeclareMathOperator*{\argmin}{arg\,min}

\newtheorem{exm}{Example}[section]

\newtheorem{prp}{Proposition}[section]

\newenvironment{proof}[1][Proof]{\begin{trivlist}
\item[\hskip \labelsep {\bfseries #1}]}{\end{trivlist}}

\newcommand{\qed}{\nobreak \ifvmode \relax \else
      \ifdim\lastskip<1.5em \hskip-\lastskip
      \hskip1.5em plus0em minus0.5em \fi \nobreak
      \vrule height0.75em width0.5em depth0.25em\fi}

\title{\LARGE \bf Road Pricing for Spreading Peak Travel: Modeling and Design}

\author[*]{Tichakorn Wongpiromsarn}
\author[$\dag$]{Nan Xiao}
\author[$\dag$]{Keyou You}
\author[$\dag$]{Kai Sim}
\author[$\dag$]{Lihua Xie}
\author[$\ddag$]{Emilio Frazzoli}
\author[$\ddag$]{Daniela Rus}
\affil[*]{\small Singapore-MIT Alliance for Research and Technology, Singapore}
\affil[$\dag$]{\small Nanyang Technological University, Singapore}
\affil[$\ddag$]{\small Massachusetts Institute of Technology, Cambridge, MA, USA}

\date{}

\begin{document}

\maketitle
\thispagestyle{empty}
\pagestyle{empty}

\begin{abstract}
A case study of the Singapore road network provides empirical evidence that road pricing can significantly affect commuter trip timing behaviors.
In this paper, we propose a model of trip timing decisions that reasonably matches the observed commuters' behaviors.
Our model explicitly captures the difference in individuals' sensitivity to price, travel time and early or late arrival at destination.
New pricing schemes are suggested to better spread peak travel and reduce traffic congestion.
Simulation results based on the proposed model are provided in comparison with the real data for the Singapore case study.
\end{abstract}

\vspace{3mm}
\section{INTRODUCTION}\label{sec:intro}
\vspace{2mm}

Traffic congestion causes significant efficiency losses, wasteful
energy consumption and excessive air pollution.
This problem arises in many urban areas because of the continual growth
in motorization and the difficulties in increasing road capacity due
to space limitations and budget constraints. As a result, traffic
management that aims at maximizing the efficiency and effectiveness
of road networks without increasing road capacity becomes
increasingly crucial.
In the recent decades, the technology in communication, control and information areas has advanced substantially,
making it possible to create intelligent traffic systems of high efficiency \cite{Orosz10TrafficJams}.

Typical strategies that aim at reducing traffic congestion include
ramp metering at freeway on-ramps, variable speed limits on freeways
and signal timing plan at signalized intersections
\cite{Kurzhanskiy10Active}. A case study on the traffic system in
California shows that transportation pricing such as congestion
pricing, parking pricing, fuel tax pricing, vehicle miles of travel
fees and emissions fees can better manage the transportation system
to a great extent \cite{Deakin96nPricing}. As another example, the
Electronic Road Pricing (ERP) system in Singapore charges motorists
when they use certain roads during the peak hours in order to
maintain an optimal speed range for both expressways and arterial
roads \cite{MenonERP00,olszewski2005modelling}.
A comprehensive review of the design and evaluation of road pricing schemes
can be found, for example, in \cite{Button98,Tsekeris09}.

The road pricing system is typically implemented for two main objectives.
First, it is designed to affect the route-choice behaviors. For
example, the charges on expressways motivate the motorists to use
alternative, less congested, arterial roads even though it comes at
the cost of extra travel time. Second, road pricing is enforced on
many of the roads in the city area in order to refrain the motorists
from using those roads during the peak hours as no alternative route
with cheaper rate is possible. Hence, a significant portion of the
motorists will either turn to public transportation or rearrange
their schedules to avoid entering the city during the peak hours.
Previous studies have mainly focused on the first objective. The
notion of Wardrop equilibrium, with travel time being the main
component in the travel cost, has been utilized in order to find a
pricing scheme that moves the user equilibrium (where all travelers
minimize their own travel cost) to the system optimum (where the
total travel time in the transportation system is minimized)
\cite{wardrop52,patriksson1994traffic,como:robust}.

In this paper, the latter objective is considered where the trip departure/arrival
time, instead of the path choice, is the decision to be made by the
motorists. We focus on modeling the effect of road pricing on
motorists' trip timing behaviors and designing the road pricing
strategy to spread peak travel and to avoid congestion.
The multinomial logit (MNL) model \cite{mcfadden1973conditional}, which is
a typical discrete choice model, has been employed, for example, in
\cite{olszewski2005modelling,chin1990influences} to study the trip re-timing behaviors.
In those studies, however, only the effect of a given pricing scheme was analyzed and the analysis was
only for the case where the motorists have a finite number of choices of departure/arrival time.
In addition, with the MNL model, the variation in the parameters of the utility function was not explicitly captured.

The main contribution of this paper is twofold.
First, we explicitly model variation in the parameters of the utility or cost function among different motorists.
Second, the traffic pricing design that aims at spreading peak travel is addressed. 

The remainder of the paper is organized as follows.
A trip timing model as well as the case study of Singapore road network are presented in the next section.
Section \ref{sec:discrete} presents a pricing strategy that better spreads peak travel.
Simulation results are provided in Section \ref{sec:case-study}.
Finally, Section \ref{sec:conclusion} concludes the paper and discusses future work.


\section{TRIP TIMING AND A MOTIVATIONAL EXAMPLE}\label{sec:motivational}

We consider a particular road $\R$ during the period of interest
and assume that each motorist decides his/her arrival time at $\R$
by minimizing his/her travel cost.
%
In general, the travel cost is different for different motorists and depends on many factors such
as one's preferred arrival time (e.g., worker's official work hours
or child's school hours) and arrival time flexibility, travel time,
road price and sensitivity to price (affected by occupational and
family status), car occupancy, transportation mode flexibility, etc.

Consider a motorist whose travel cost of arriving at road
$\R$ at time $t$ is defined by $J(t)$.
As noted earlier, $J(t)$ may be different for different motorists.
The optimal arrival time at road $\R$ of this motorist is given by
\begin{equation}
\label{eq:model}
t^*=\arg\min_{t}J(t).
\end{equation}

Without pricing, there would be a high concentration of demand
during the rush hour, leading to congestion. 
%
As a motivating example,
consider the Tanjong Pagar area (Fig.~\ref{fig:tanjong-pagar}), which is
located in the heart of the central business district (CBD) of Singapore.
From the locations of the ERP gantries and the directions
of the roads, it can be checked that the motorists get charged the
same rate during the peak hours no matter which road they pick to
enter this area.
\begin{figure}[tbh]
\centering
\includegraphics[scale=0.6]{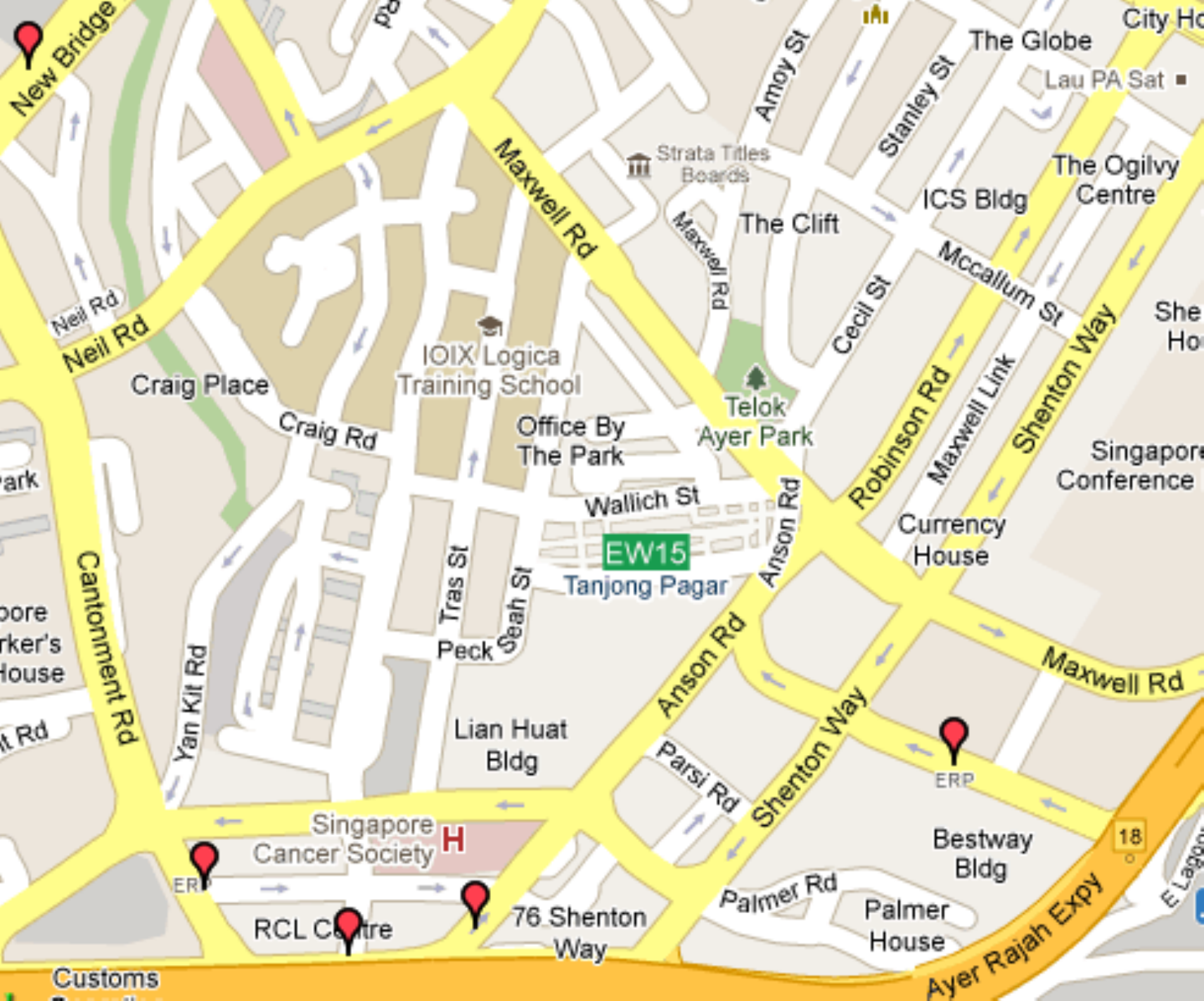}
\caption{The map of Tanjong Pagar area, Singapore with the locations
of the ERP gantries.} \label{fig:tanjong-pagar}
\end{figure}

Fig.~\ref{fig:anson-flow} shows the traffic flow on Anson Road,
which is one of the roads that can be used to enter the Tanjong Pagar area,
during the weekdays of August 2010 as well as the ERP rate. For the
morning peak hours (roughly, from 7am to 10:30am), there are 3
noticeable peaks in the flow: at 7:55am, which is right before the
ERP is effective, around 8:45am when the charge is maximum and at
10:00am, which is right after the ERP become inactive. From this
consistent observation over all the weekdays of August 2010, it is
reasonable to conclude that a significant portion of the motorists
who regularly use this road intentionally adjust their schedule to
avoid being charged.
\begin{figure}[tbh]
\centering
\includegraphics[trim=3.5cm 0cm 1.2cm 0cm, clip=true, width=0.7\textwidth]{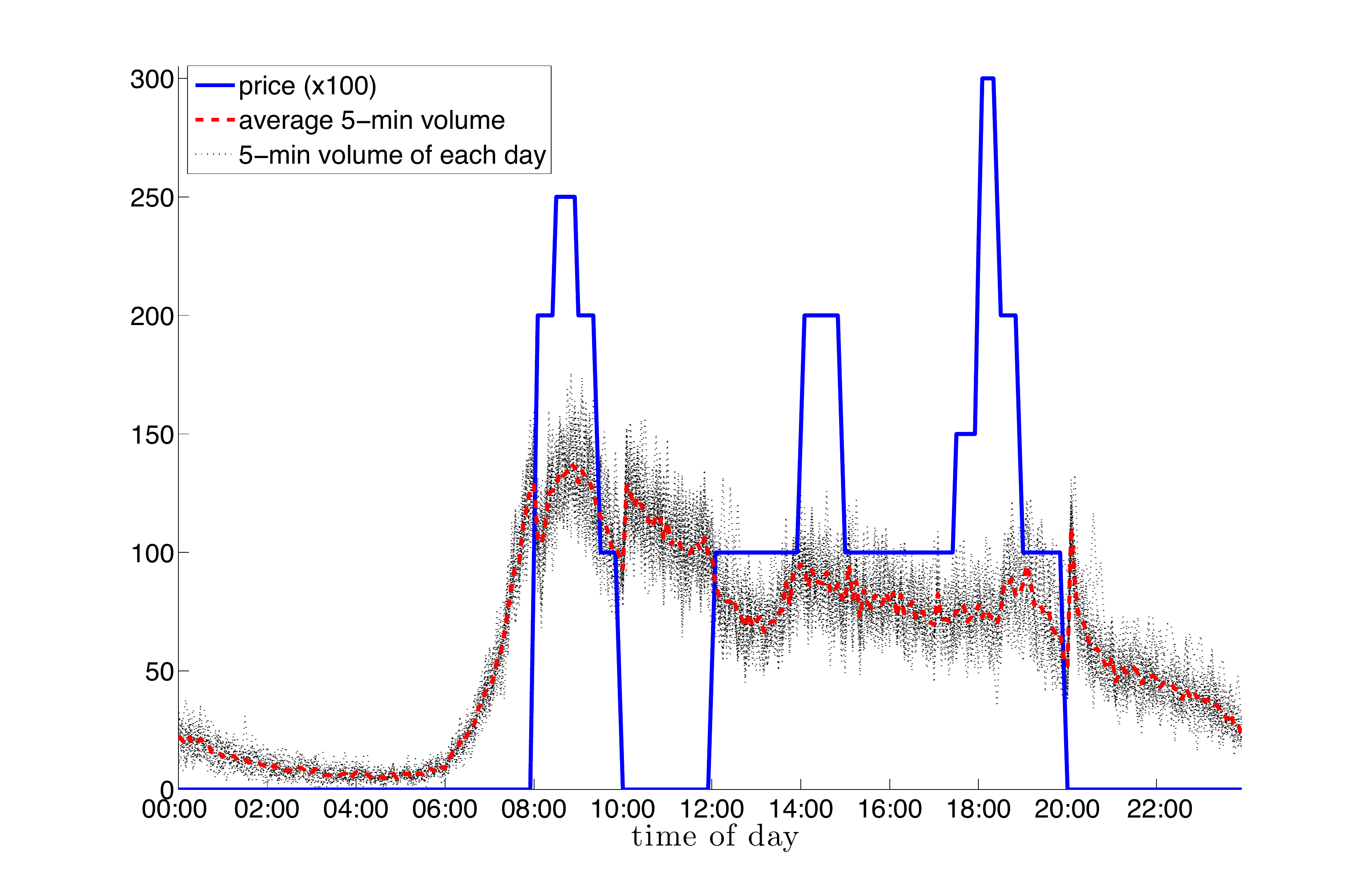}
\caption{Traffic flow on Anson Road during the weekdays of August
2010.} \label{fig:anson-flow}
\end{figure}

Next we propose a model that explains the behavior observed in
Fig.~\ref{fig:anson-flow}.
Let $\T \subset \reals_{\geq 0}$ denote the domain of $t$,
$p(t), t \in \T$ denote the road price at time $t$,
and $d(t), t \in \T$ denote the expected travel time from the motorist's origin to destination,
assuming that the motorist arrives at road $\R$ at time $t$.
Consider the case where the travel cost is given by
\begin{equation}
\label{eq:travel-cost}
J(t) = p(t) + J_t(t) + J_d(d(t)),
\end{equation}
where $J_t: \T \to \reals_{\geq 0}$ captures the cost of arriving earlier or later
than the preferred arrival time $T$ and
$J_d: \reals_{\geq 0} \to \reals_{\geq 0}$ captures the travel time factor of the travel cost.

As an initial step, we neglect the travel time factor and let
\begin{equation}
\label{eq:Jt}
\begin{array}{rcl}
J_d(d) &=& D, \forall d \geq 0\\
J_t(t) &=& \left\{ \begin{array}{ll}
b_1(T - t) &\hbox{if } t \leq T\\
b_2(t - T) &\hbox{otherwise}
\end{array}\right. ,
\end{array}
\end{equation}
where $D \geq 0$ is a constant
and $b_1, b_2 \geq 0$ are the parameters that represent the
amount of money the motorist is willing to pay to save a minute of early and late arrival respectively
and may be different for different motorists.
Assume that the preferred arrival time $T$ of each motorist is 8:45 am.
Fig.~\ref{fig:opttime} shows the optimal arrival time for
each value of $b_1$ and $b_2$ for the case where $p(t)$ is the
current price implemented on Anson road (cf.
Fig.~\ref{fig:anson-flow}).

\begin{figure}[tbh]
\centering
    \includegraphics[trim=3.5cm 0cm 1.5cm 0cm, clip=true, width=0.4\textwidth]{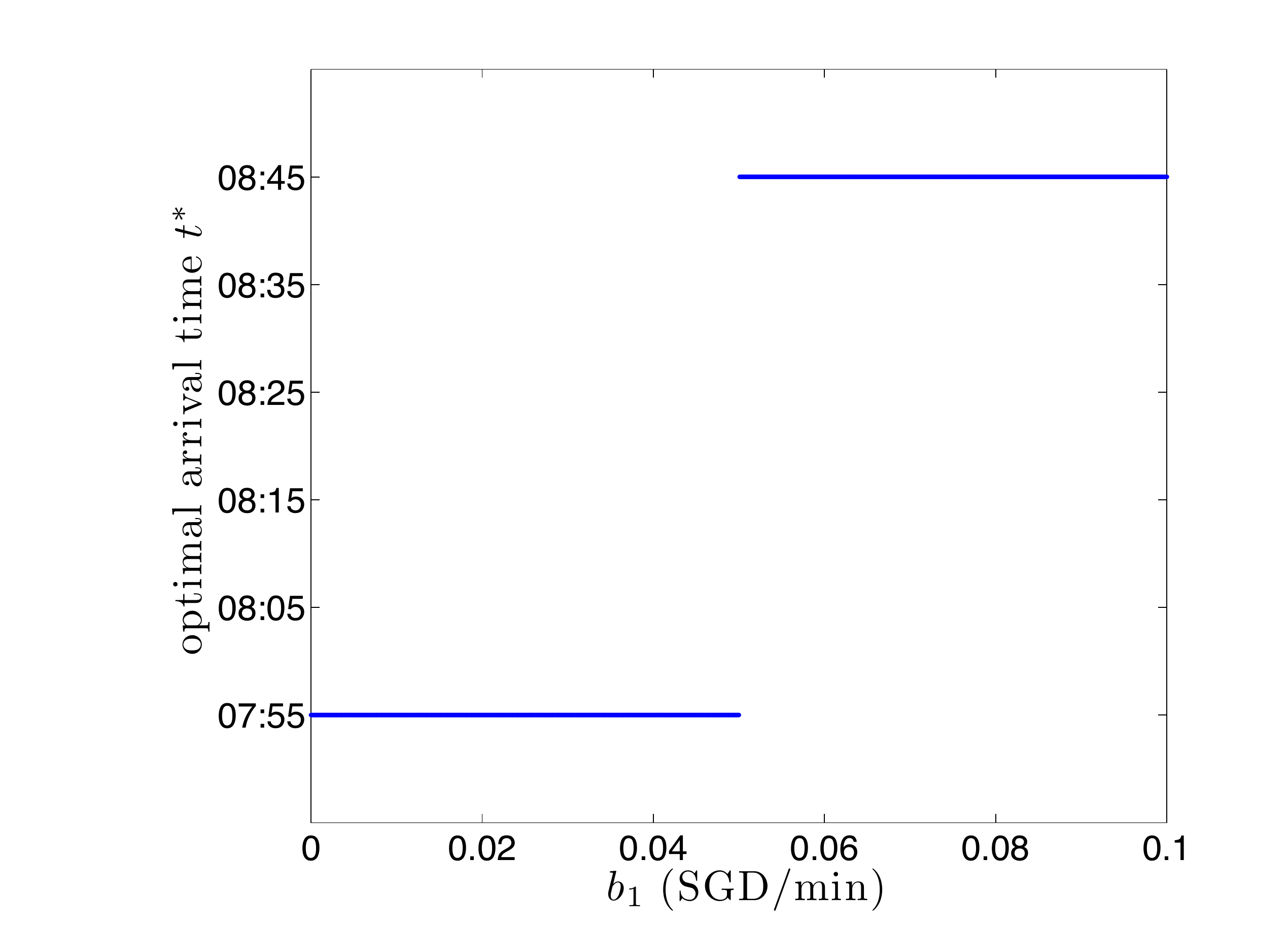}
    \hfill
    \includegraphics[trim=3.5cm 0cm 1.5cm 0cm, clip=true, width=0.4\textwidth]{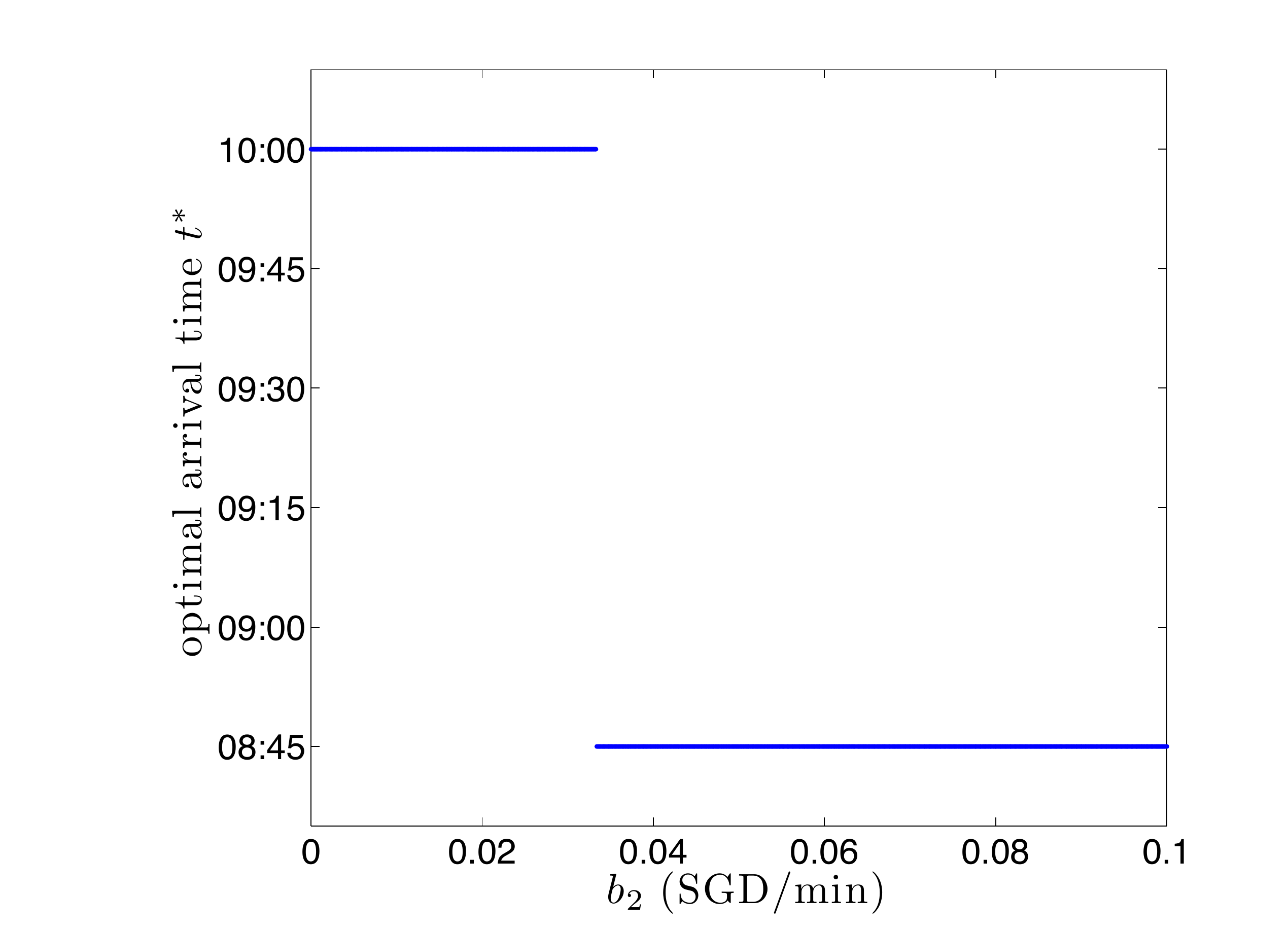}
\caption{The $b$-$t^*$ relations, i.e., the optimal arrival time
$t^*$ for each value of $b_1$ and $b_2$ for the case where $p(t)$ is
the current price implemented on Anson road.} \label{fig:opttime}
\end{figure}

From Fig.~\ref{fig:opttime}, the optimal arrival time is only one of
the followings: (a) 7:55am, which is right before the ERP is
effective, (b) 8:45am, which is the preferred arrival time, (c)
10:00am, which is right after the ERP become inactive. This matches
the observed behaviors of the motorists in
Fig.~\ref{fig:anson-flow}. However, this optimal arrival time
distribution is not efficient as there is a high concentration of
demand only at 3 different times. Ideally, the optimal arrival times
should be distributed equally among various time slots. In the next
sections, we derive a pricing scheme that results in such an equally
distributed optimal arrival times.

\section{PRICING STRATEGIES}\label{sec:discrete}



%



As a starting point, we consider the cost function in
(\ref{eq:travel-cost}) with $J_d(d)$ and $J_t(t)$ as defined in (\ref{eq:Jt}).
Assume that the preferred arrival time $T$ is the same for all motorists.
Then, the optimal arrival time $t^*$ for
each motorist with respect to the travel cost $J(t)$ only depends on
the pricing scheme $p(t)$ and his/her time-money trade-off
parameters $b_1$, $b_2$. For the simplicity of the presentation, we
consider only motorists who prefer early over late arrival, i.e.,
$b_2 \gg b_1$, and refer to $b_1$ simply as $b$ for the rest of this
section. Similar results can be derived for motorists who prefer
late over early arrival.

For a given pricing scheme $p$, we define a map $F_p : \reals_{\geq
0} \to [0, T]$ that takes the time-money trade-off parameter and
returns the maximum optimal arrival time as follows%
\footnote{For certain values of $b$, $J(t)$ may attain its minimum value at two different values of $t$.
In this case, we simply assume that the motorist picks the maximum of such optimal arrival times,
i.e., the arrival time that is closest to his desired arrival time, as his/her actual arrival time.}
\begin{equation}
\label{eq:Fp}
  F_p(b) = \max \left( \argmin_t \big(p(t) + b(T - t) + D \big) \right).
\end{equation}

Given a desired map $F : \reals_{\geq 0} \to [0, T]$, in this
section, we derive a pricing scheme $p$ such that $F_p = F$. We
consider the case where $F$ is a monotonically increasing step
function, e.g., as shown in Fig.~\ref{fig:desired-b-t}. In this
case, $F$ can be written as
\begin{equation}
\label{eq:F}
  F(b) = \left\{ \begin{array}{ll}
  T_i &\hbox{if } B_i \leq b < B_{i+1}, \forall i \in \{1, \ldots, N-1\}\\
  T_N &\hbox{if } b \geq B_N
\end{array}\right. ,
\end{equation}
where $N \in \naturals$, $0 = B_1 < B_2 < \ldots < B_N$ and $0 \leq T_1 < T_2 < \ldots < T_N =
T$.


\begin{figure}[tbh]
\centering
    \includegraphics[width=0.7\textwidth]{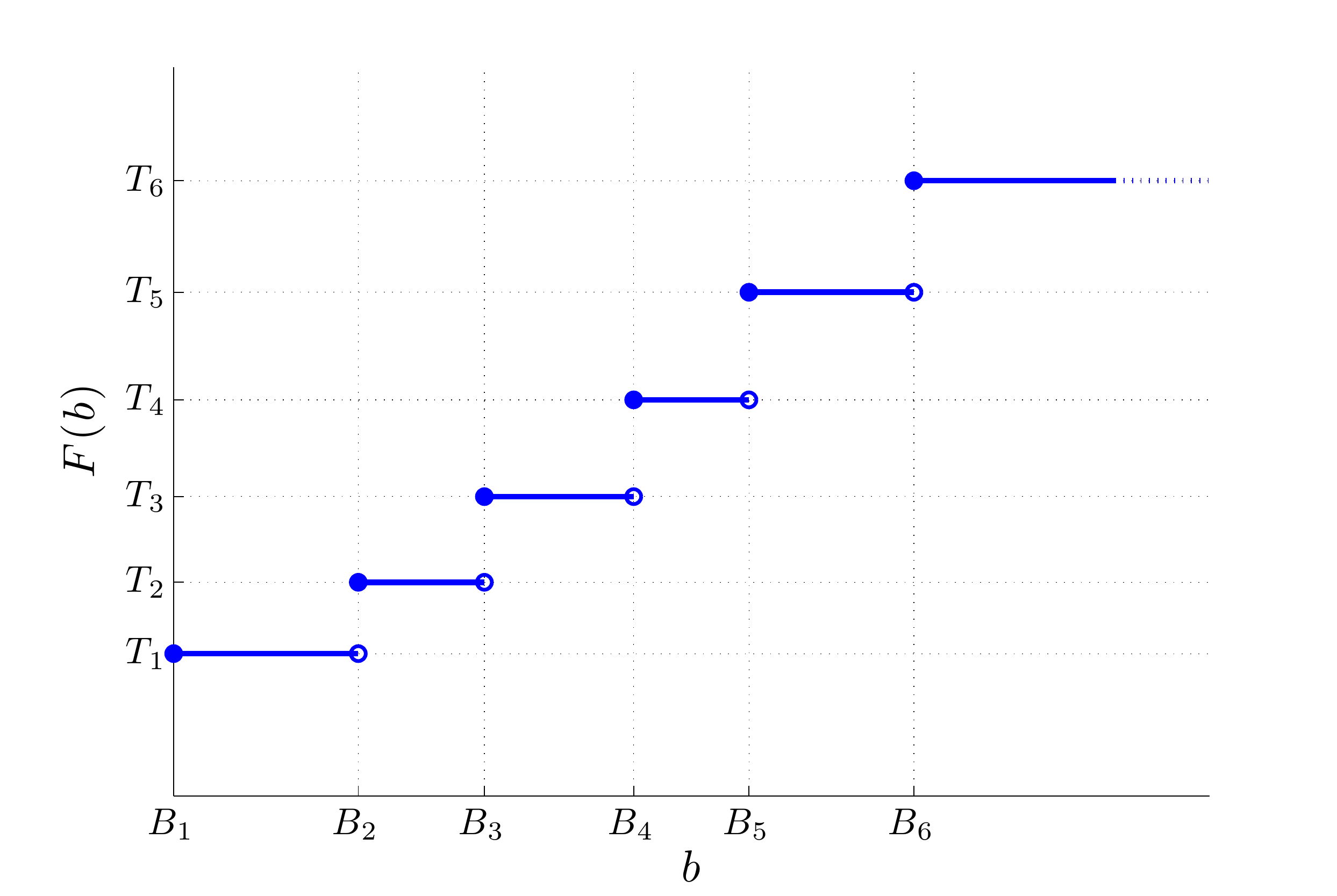}
\caption{Example of a desired map $F$ for $N = 6$.} \label{fig:desired-b-t}
\end{figure}


\begin{prp}
Consider the case where $p(t)$ is a step function. Define $p_1, p_2,
\ldots, p_N$ such that
\begin{equation}
\label{eq:price} p(t) = \left\{ \begin{array}{ll}
  p_1 &\hbox{if } t \leq T_1\\
  p_i &\hbox{if } T_{i-1} < t \leq T_i, \forall i \in \{2, \ldots, N\}
\end{array}\right.
\end{equation}
Then, $F_p = F$ if
\begin{equation}
\label{eq:pi}
  p_i = \sum_{k=2}^i B_k(T_k - T_{k-1}) + p_1, \forall i \in \{2, \ldots, N\}.
\end{equation}
\end{prp}
\begin{proof}
From (\ref{eq:Fp}), (\ref{eq:F}) and (\ref{eq:price}), it can be checked that a necessary and sufficient condition
for $F_p = F$ is that for each $i,j \in \{1, \ldots, N\}$,
\begin{equation}
  \label{pf:cond-fixedT}
  p_i + b(T - T_i) \leq p_j + b(T - T_j), \forall b \in [B_i, B_{i+1}).
\end{equation}
The case where $j = i$ is trivial so we only need to consider the case where $j \not= i$.
Since $T_j > T_i$ and $B_j > B_i \geq 0$ for all $j > i$, condition (\ref{pf:cond-fixedT}) is satisfied if
\begin{equation}
\label{pf:cond-fixedT2}
  p_i \leq \left\{ \begin{array}{ll}
  B_i(T_i - T_j) + p_j, &\forall j < i\\
  B_{i+1}(T_i - T_j) + p_j, &\forall j > i
  \end{array}\right.
\end{equation}

First, consider the case where $j < i$.
Since $B_k > B_i, \forall k > i$, we get
\begin{equation*}
  \sum_{k = j+1}^i B_k(T_k - T_{k-1}) \leq \sum_{k=j+1}^iB_i(T_k - T_{k-1}) = B_i(T_i - T_j).
\end{equation*}
Adding $\sum_{k=2}^j B_k(T_k - T_{k-1}) + p_1$ to both sides and using (\ref{eq:pi}), we get
\begin{eqnarray*}
  p_i &=& \sum_{k = 2}^i B_k(T_k - T_{k-1}) + p_1 \\
         &\leq& B_i(T_i - T_j) + \sum_{k=2}^j B_k(T_k - T_{k-1}) + p_1\\
         &=& B_i(T_i - T_j) + p_j
\end{eqnarray*}

With similar reasoning, for the case where $j > i$, we have
\begin{equation*}
  B_{i+1}(T_j - T_i) = \sum_{k=i+1}^j B_{i+1}(T_k - T_{k-1})  \leq \sum_{k = i+1}^j B_k(T_k - T_{k-1}).
\end{equation*}
Adding $B_{i+1}(T_i - T_j) + \sum_{k=2}^i B_k(T_k - T_{k-1}) + p_1$ to both sides, we get
\begin{eqnarray*}
  p_i &=& \sum_{k = 2}^i B_k(T_k - T_{k-1}) + p_1 \\
         &\leq& B_{i+1}(T_i - T_j) + \sum_{k=2}^j B_k(T_k - T_{k-1}) + p_1\\
         &=& B_{i+1}(T_i - T_j) + p_j
\end{eqnarray*}

Hence, condition (\ref{pf:cond-fixedT2}) is satisfied and we can conclude that $F_p = F$.

\end{proof}

\begin{exm}
Consider the case where $N = 6$, $B_i = 0.02(i-1)$,
$T_i = T - 10(N-i)$, $i \in \{1, \ldots, N\}$.
With $p_1 = 0$, according to (\ref{eq:pi}), we get
$p_2 =  0.2$, $p_3 = 0.6$, $p_4 = 1.2$, $p_5 = 2$ and $p_6 = 3$.
This pricing scheme as well as the one currently implemented are shown in Fig.~\ref{fig:newprice}.
\end{exm}

\begin{figure}[tbh]
\centering
   \includegraphics[trim=1.5cm 0cm 1.5cm 0cm, clip=true, width=0.7\textwidth]{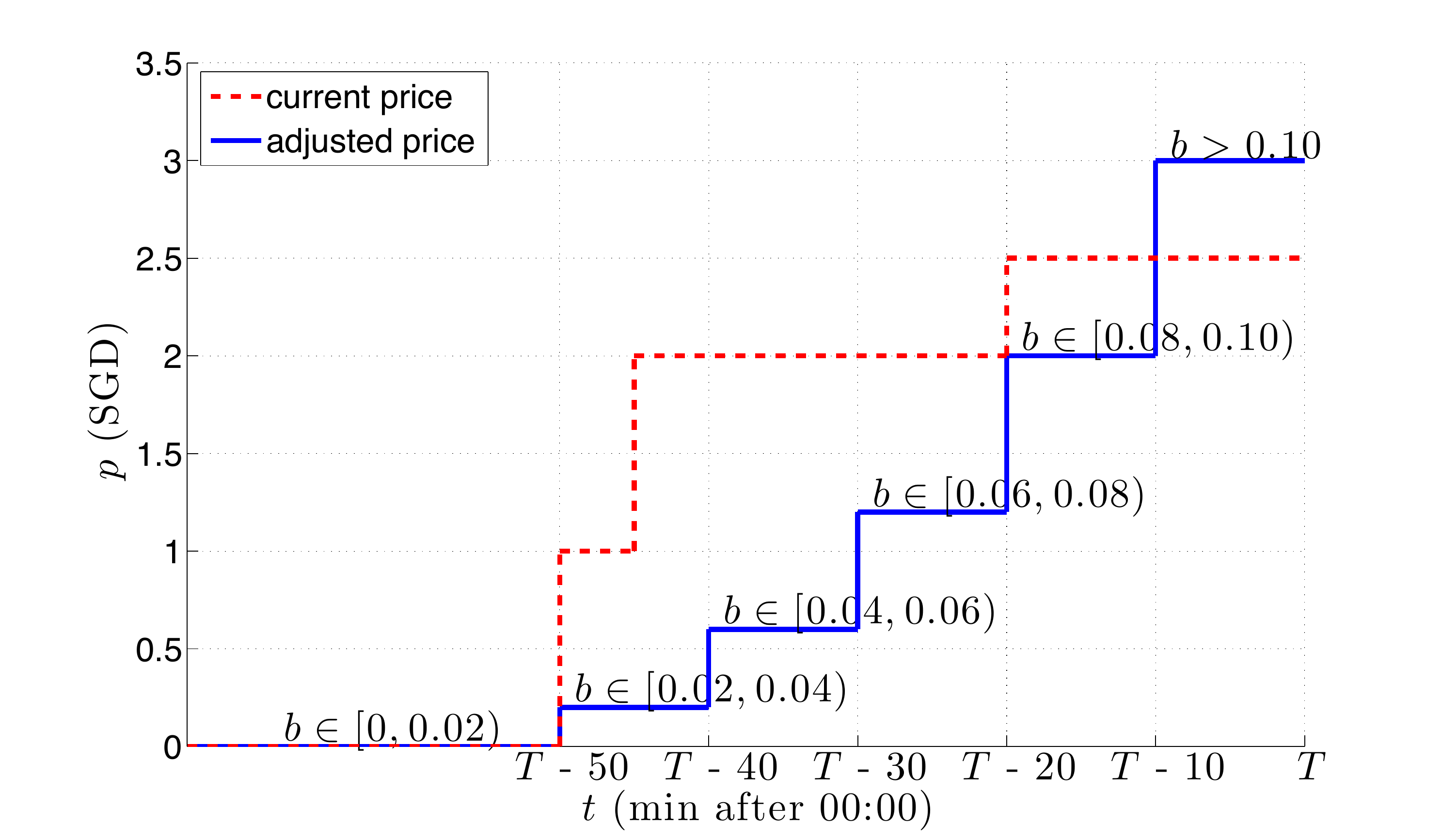}
\caption{Example of the proposed pricing scheme. The range of $b$ associated with each price is also shown.}
\label{fig:newprice}
\end{figure}



%
%

\section{CASE STUDY OF SINGAPORE}
\label{sec:case-study}

Reconsider the motivational example in
Section~\ref{sec:motivational}. Assume that $b$ has a Gaussian
distribution with a certain mean $\mu_b$ and variance $\sigma_b$.
The set of possible means and variances of b can be computed from
historical data. For example, Fig.~\ref{fig:bdistr} shows possible
distributions of $b$ based on the ratio between the average number
of motorists at 7:55am and at 8:45am (cf.
Fig.~\ref{fig:anson-flow}).

Using the distribution with mean $0.051$ (Fig.~\ref{fig:bdistr}),
the map $F$ can be computed such that the numbers of motorists at
times $T_1, \ldots, T_N$ are equal. An example of such $F$ for $N =
6$, $T_i = T - 10(N-i), i \in \{1, \ldots, N\}$ is
shown in Fig.~\ref{fig:bdistr205}. The corresponding pricing scheme
based on Eq (\ref{eq:pi}) is shown in Fig.~\ref{fig:price_mu205}.

\begin{figure}[b]
\centering
   \includegraphics[width=0.4\textwidth]{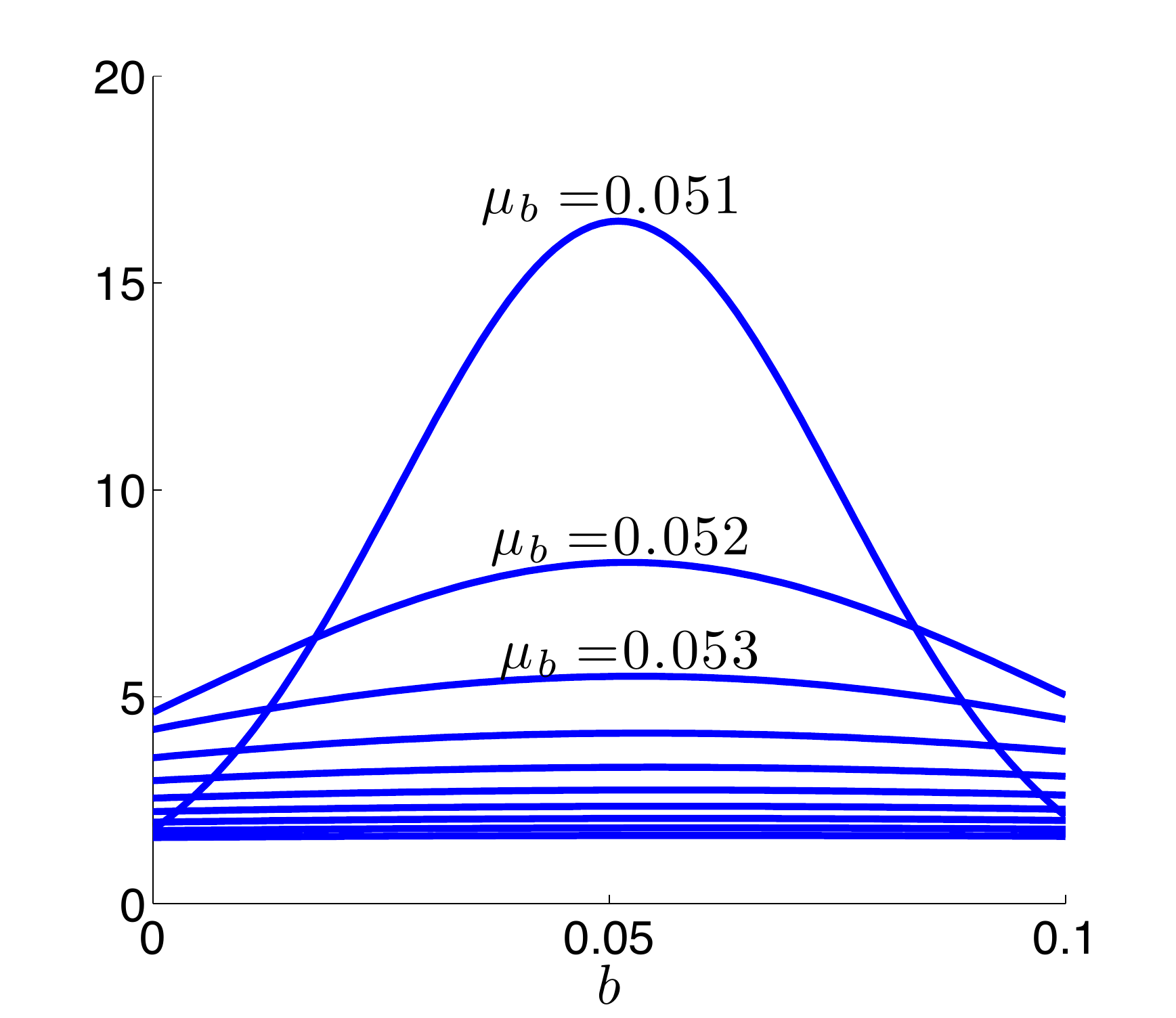}
   \hfill
   \includegraphics[width=0.4\textwidth]{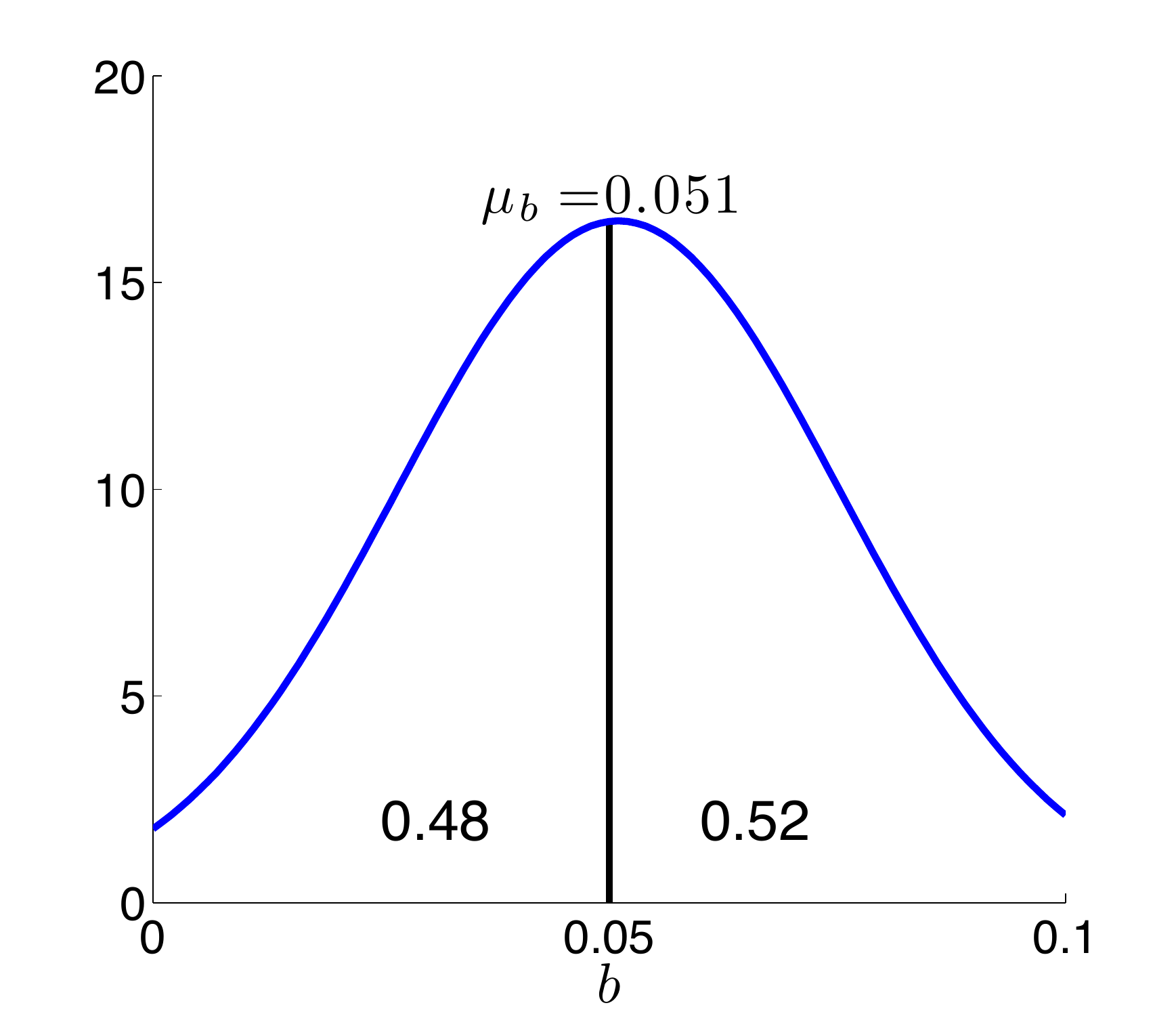}
\caption{(left) Possible distributions of $b$. (right) Distribution of $b$ with mean $0.051$} \label{fig:bdistr}
\end{figure}

\begin{figure}[tbh]
\centering
   \includegraphics[trim=1.2cm 0cm 2cm 0cm, width=0.7\textwidth]{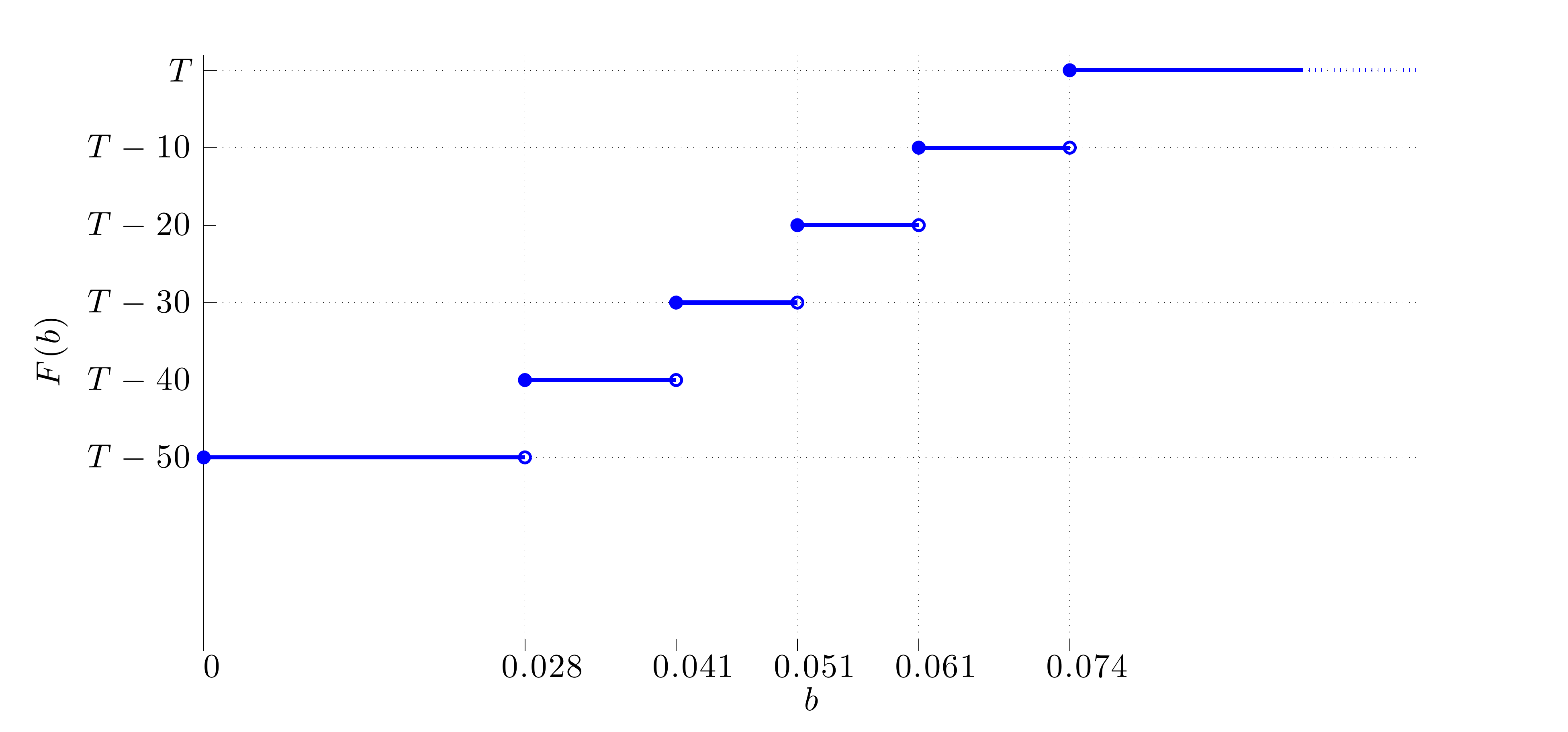}
\caption{A map $F$ that leads to the equal number of motorists at times $T_1, \ldots, T_6$,
assuming that $b$ has a Gaussian distribution with mean 0.051.} \label{fig:bdistr205}
\end{figure}

\begin{figure}[tbh]
\centering
   \includegraphics[width= 0.7\textwidth]{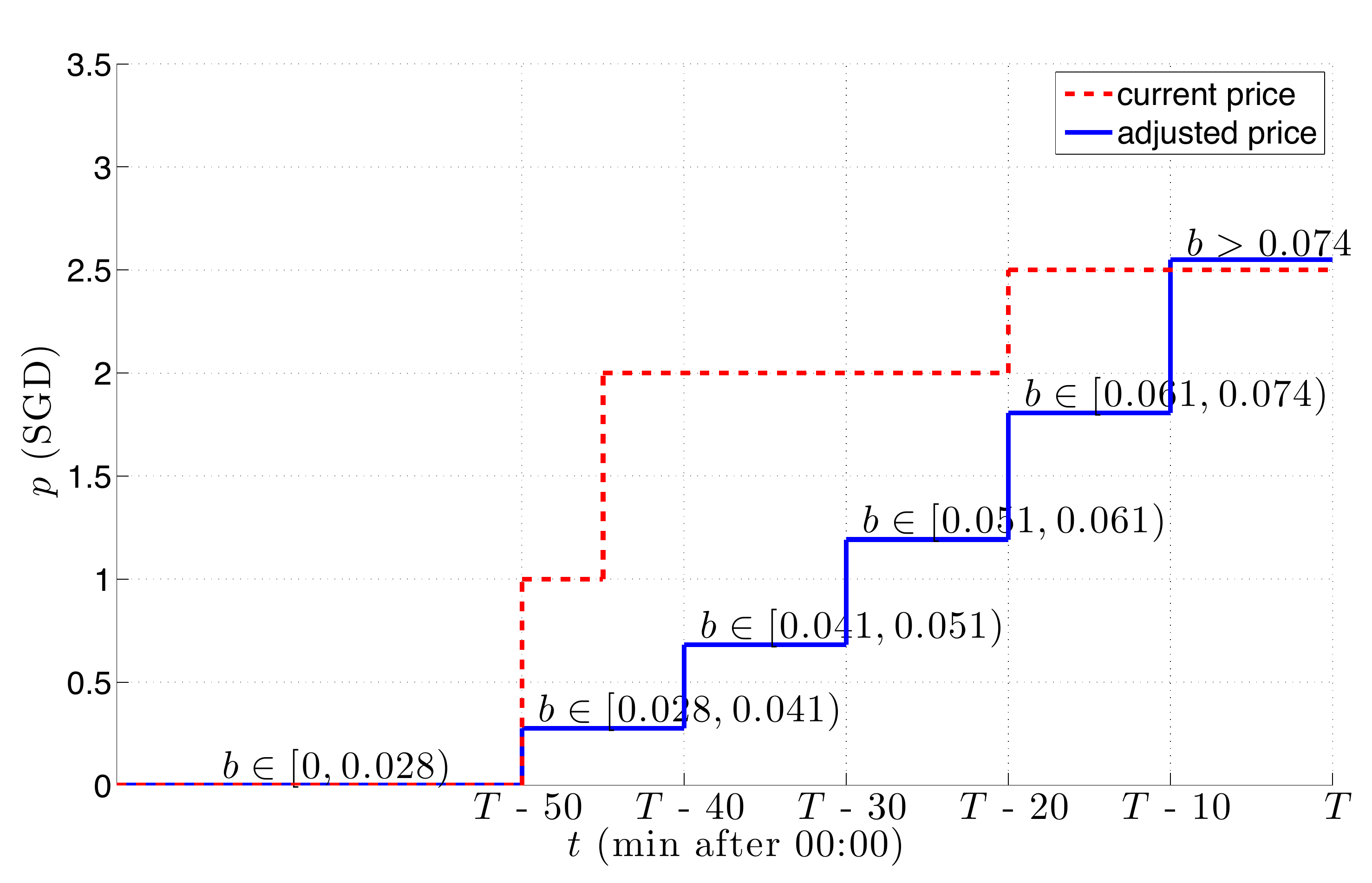}
\caption{The corresponding pricing scheme based on the map $F$
provided in Fig.~\ref{fig:bdistr205}.} \label{fig:price_mu205}
\end{figure}

The current flow between 7:55am and 8:45am and the results from 10,000,000 Monte Carlo simulations
based on the model in (\ref{eq:model}) with the cost function $J(t)$ defined in (\ref{eq:travel-cost}) and (\ref{eq:Jt})
are shown in Fig.~\ref{fig:simresults}.
The actual arrival time is obtained by adding some Gaussian noise to the optimal arrival time $t^*$.

\begin{figure}[tbh]
\centering
   \includegraphics[width= 0.7\textwidth]{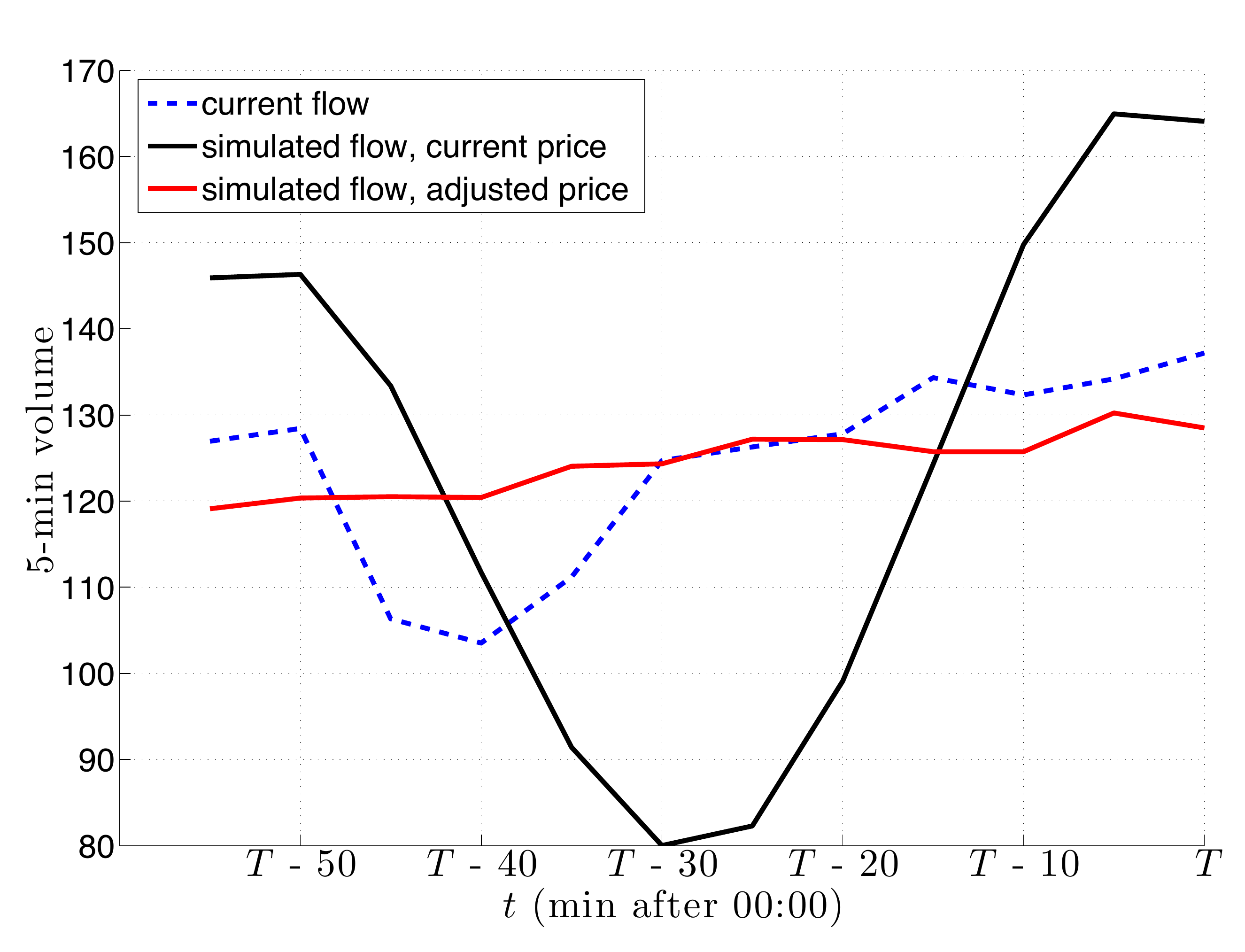}
\caption{The current flow between 7:55am and 8:45am and the results 
obtained from 10,000,000 Monte Carlo simulations 
using the model in (\ref{eq:model}) with the current price and the adjusted price based on Eq (\ref{eq:pi}).} 
\label{fig:simresults}
\end{figure}

\section{CONCLUSIONS AND FUTURE WORK}
\label{sec:conclusion}
We provided a case study of Singapore road network that shows that road pricing could significantly affect commuter trip timing behaviors.
Based on this empirical evidence, we proposed a model that describes the commuter trip timing decisions.
The analysis and simulation results showed that the proposed model reasonably matches the observed behaviors.
In addition, we proposed a pricing scheme based on the proposed model in order to
better spread peak travel and reduce traffic congestion.
Simulation results showed that uniform distribution of arrival times among motorists who regularly use the road during the peak hours could be obtained.

Future work includes considering multiple roads and incorporating the route choice behavior in the model.
We also plan to take into account stochasticity in the actual arrival time as the motorist may not arrive exactly at his/her optimal arrival time.
In addition, we are interested in incorporating the travel time factor in the model.
The average travel time for different origins and destinations as well as the portion of motorists with those origin and destination pairs
can be estimated from the data obtained from all the taxi trips that went through a road of interest.
We also plan to introduce stochasticity in the preferred arrival time $T$ (which depends, for example, on individuals' work hours).

\section{ACKNOWLEDGMENTS}
The authors gratefully acknowledge Ketan Savla for the inspiring discussions 
and Land Transport Authority of Singapore for providing the data collected from the loop detectors
installed on Anson road.
This work is supported in whole or in part by the Singapore National Research Foundation (NRF) through the Singapore-MIT Alliance for Research and Technology (SMART) Center for Future Urban Mobility (FM).


\bibliographystyle{ametsoc}
\bibliography{ref}

\end{document}